\documentclass[a4paper,12pt]{article}

\usepackage{amsmath,amssymb,mathrsfs,cite}
\usepackage{amsthm,bm,
color}
\usepackage[dvipdfmx]{graphicx}

\makeatletter
\def\figcaption{\def\@captype{figure}\caption}
\makeatother

\makeatletter
 
  \@addtoreset{equation}{section}
 \makeatother

\newtheorem{theorem}{\bf Theorem}[section]
\newtheorem{lemma}[theorem]{\bf Lemma}
\newtheorem{proposition}[theorem]{\bf Proposition}
\newtheorem{corollary}[theorem]{\bf Corollary}

\newtheorem{example}{\bf Example}[section]
\newtheorem{remark}{\bf Remark}[section]
\newtheorem{definition}{\bf Definition}[section]

\newcommand{\ro}[1]{\expandafter{\romannumeral#1}}
\newcommand{\Ro}[1]{\uppercase\expandafter{\romannumeral#1}}

\title{Essential spectrum of the discrete Laplacian\\ on a perturbed periodic graph}
\author{
Itaru Sasaki\thanks{Department of Mathematical Sciences,
	Faculty of Science, Shinshu University, Asahi, Matsumoto 390-8621, Japan,
	email: isasaki@shinshu-u.ac.jp}
and 
Akito Suzuki
\thanks{Division of Mathematics and Physics, 
	Faculty of engineering, Shinshu University, Wakasato, Nagano 380-8553, Japan,
	e-mail: akito@shinshu-u.ac.jp}
}
\begin{document}

\maketitle
\begin{abstract}
We address the Laplacian on a perturbed periodic graph
which might not be a periodic graph.
We present a class of perturbed graphs for which
the essential spectra of the Laplacians are stable
even when the graphs are perturbed 
by adding and removing infinitely many vertices and edges. 
Using this result, 
we demonstrate how to determine the spectra of 
cone-like graphs, the upper-half plane, 
and graphs obtained from $\mathbb{Z}^2$ by randomly adding vertices. 
\end{abstract}
\section{Introduction} 

The spectral properties of the Laplacians and Scr\"odinger operators on periodic graphs
have been studied by many authors 
\cite{A, AnIsMo, HiNo, HiShi, HSSS, IsKo, IsMo, KoSa, RaRo, Su} 
(see also the references therein).
In this paper, the essential spectrum of the Laplacian on a perturbed periodic graph is considered. 
It is well-known that if the perturbation is compact,
the essential spectrum is stable (see Proposition \ref{prop2.1}).
We are interested in the case in which the perturbation is possibly non-compact,
{\rm i.e.}, the operator ``$L_{G^\prime} - L_G$" is not a compact operator, 
where $L_G$ (resp., $L_{G^\prime}$) is the Laplacian on a periodic graph $G$ 
(resp., a perturbed graph $G^\prime$ of $G$). 
If $G^\prime$ is a graph  obtained from $G$ by removing and adding some vertices,
then  $G$ is not a subgraph of $G^\prime$, and {\it vice versa}.
In such cases, the meaning of ``$L_{G^\prime} - L_G$" is unclear,
because $L_G$ and $L_{G^\prime}$ act on different Hilbert spaces.
The precise meaning of ``$L_{G^\prime} - L_G$" is given in \eqref{perturb}.
It is noteworthy that a perturbed periodic graph $G^\prime$ might not be periodic.
In general, it is difficult to determine the spectrum of an infinite graph,
if it does not have a nice symmetry, such as periodicity.
In this paper,
we present a class of perturbed periodic graphs $G^\prime$
such that the essential spectrum of an unperturbed graph $G$ 
is contained in that of $G^\prime$:
\begin{equation}
\label{Eq1.1} 
\sigma_{\rm ess}(L_G) \subset \sigma_{\rm ess}(L_{G^\prime}). 
\end{equation}
We emphasize that the converse of \eqref{Eq1.1}
cannot be expected in general.
As shown in Example \ref{counter_ex},
there exists a perturbed periodic graph $G^\prime$ such that
$\sigma_{\rm ess}(L_G) \subset \sigma_{\rm ess}(L_{G^\prime})$
and $\sigma_{\rm ess}(L_{G^\prime}) \setminus  \sigma_{\rm ess}(L_G) \not=\emptyset$. 

In our definition \eqref{DefLG}, 
$L_G$ is self-adjoint, and its spectrum $\sigma(L_G)$ is contained in $[-1,1]$.
This property raises the question of whether 
$\sigma(L_G)$ is already all of $[-1,1]$. 
In their paper \cite{HiShi}, 
Higuchi and Shirai stated that
$G$ has the full spectrum property (FSP) if $\sigma(L_G) = [-1,1]$,
and studied the problem of whether an infinite graph $G$ has the FSP.
If \eqref{Eq1.1} holds and $G$ has the FSP,
then $G^\prime$ has the FSP (see Corollary \ref{FSP}). 
Therefore, it is possible to determine the spectra of perturbed graphs of $\mathbb{Z}^d$,
such as those of cones (Example \ref{ex_cone}) and the upper-half plane (Example \ref{ex_upp}).
We also discuss the spectrum of a graph obtained from $\mathbb{Z}^2$
by randomly adding pendants.

This paper is organized as follows.
In Section 2,
we present some basic facts on 
perturbed and periodic graphs.
Section 3 is devoted to the study of 
the essential spectra of perturbed periodic graphs.
In Section 4,
we demonstrate how to determine
the spectrum of a perturbed periodic graph,
using the results established in Section 3.
We present the proofs of technical lemmas
in the appendix. 
\section{Preliminaries}
Let $G = (V(G), E(G))$ be an unoriented graph (possibley having 
loops and multiple edges), 
where $V(G)$ and $E(G)$ are the sets of vertices and unoriented edges, respectively.
We use an ordered pair $\{x, y\} \in V(G) \times V(G)$ 
to denote the endpoints of an edge $e \in E(G)$
and then write $V(e) = \{x, y\}$.
We consider that each edge $e \in E(G)$ has two orientations
and introduce the set $A(G)$ of all oriented edges $e$,
whose origins and terminals are denoted by $o(e)$ and $t(e)$,
respectively. 
We denote the set of all oriented edges whose origin is $x$ by
\[ A_x(G) = \{ e \in A(G) \mid o(e) = x \} \] 
and the number of all edges in $A_x(G)$ by ${\rm deg}_G x = \# A_x(G)$.
If there is no danger of confusion,
we omit $G$ in ${\rm deg}_G$.
Throughout this paper, 
unless otherwise noted,
we assume that $G$ is locally finite and
\begin{equation}
\label{deg>=1}
1 \leq \inf_{x \in V(G)} {\rm deg} x \leq \sup_{x \in V(G)} {\rm deg} x < \infty. 
\end{equation}
The left-hand side of \eqref{deg>=1} implies that there is no isolated vertex.

The Laplacian we address in this paper is defined as
\begin{equation}
\label{DefLG} 
(L_G\psi)(x) = \frac{1}{{\rm deg} x} 
	\sum_{e \in A_x(G)} \psi(t(e)), 
	\quad \psi \in \ell^2(V(G)), 
\end{equation}
where 
\[ \ell^2(V(G)) 
	= \{ \psi:V(G) \to \mathbb{C}
		\mid \langle \psi, \psi \rangle < \infty \} 
\]
is the Hilbert space with the inner product 
\[ \langle \psi, \psi \rangle
	= \sum_{x \in V(G)} |\psi(x)|^2 {\rm deg} x.
\]

We say that a graph (possibly having loops and multiple edges) $G^\prime$ 
is isomorphic to $G$ and write $G^\prime \simeq G$ if 
there exists a pair of bijections $\varphi_V:V(G^\prime) \to V(G)$ and 
$\varphi_E:E(G^\prime) \to E(G)$ such that
for all $e \in E(G)$ 
with endpoints $V(e) = \{x, y\}$,
\begin{equation}
\label{consv}
V(\varphi_E^{-1}(e)) = \{\varphi_V^{-1}(x), \varphi_V^{-1}(y)\}.
\end{equation} 
In this case, we can introduce an orientation-preserving bijection
$\varphi_A:A(G^\prime) \to A(G)$ as
$\iota^\prime(\varphi_A^{-1}(e)) = \varphi_E^{-1}(\iota(e))$
and
\begin{equation}
\label{vecon} 
o(\varphi_A^{-1}(e)) = \varphi_V^{-1}(o(e)), 
	\quad t(\varphi_A^{-1}(e)) = \varphi_V^{-1}(t(e))
		\quad (e \in A(G)), 
\end{equation} 
where $\iota:A(G) \to E(G)$
and $\iota^\prime:A(G^\prime) \to E(G^\prime)$ are natural surjections. 
We know from \eqref{vecon} that $\varphi^{-1}_A(A_x(G)) = A_{\varphi^{-1}_V(x)}(G^\prime)$
and ${\rm deg}_{G^\prime} \varphi^{-1}_V(x) = {\rm deg}_G x$ for all $x \in V(G)$.
If $G^\prime$ is isometric to $G$, 
we can define a natural unitary operator
$\mathscr{U}:\ell^2(V(G)) \to \ell^2(V(G^\prime))$ as
\begin{equation}
\label{Vecid} 
(\mathscr{U}\psi)(\varphi_V^{-1}(x)) = \psi(x), 
	\quad x \in V(G) 
\end{equation}
for $\psi \in \ell^2(V(G))$.
By \eqref{vecon} and \eqref{Vecid},
\begin{equation} 
\label{Lid}
L_{G^\prime}\mathscr{U} = \mathscr{U} L_G. 
\end{equation}
\subsection{Perturbed graph}
We say that $G_0$ is a subgraph of $G$ and write $G_0 \subset G$
if $G_0$ is a graph satisfying $V(G_0) \subset V(G)$ and $E(G_0) \subset E(G)$.
A graph $G$ with $V(G) \not= \emptyset$ is called non-empty.
\begin{definition}
\label{def1.1}
{\rm
We say that $G^\prime$ is a {\it perturbed graph} of a graph $G$
if there exist non-empty subgraphs $G_0^\prime$ of $G^\prime$
and $G_0$ of $G$ such that $G^\prime_0 \simeq G_0$. 
}
\end{definition}
\begin{remark}
\label{1.1}
{\rm
Definition \ref{def1.1} allows the case where $G \subset G^\prime$,
and {\it vise verca}. 
If $G^\prime$ is a graph obtained by adding vertices and edges to $G$,
then $G^\prime$ is a subgraph of $G$.
On the other hand, 
if $G^\prime$ is a graph obtained by removing vertices and edges from $G$,
then $G^\prime$ is a subgraph of $G$.
In cases where $G^\prime$ is a graph obtained from  $G$
by adding and removing vertices and edges,
$G$ is not a subgraph of $G^\prime$,
and {\it vise verca}. 
Such a case is also included in Definition \ref{def1.1}.
}
\end{remark}

Let $G^\prime$ be a perturbed graph of a graph $G$ with bijections 
$\varphi_V:V(G_0^\prime) \to V(G_0)$ and $\varphi_E:V(G_0^\prime) \to E(G_0)$
satisfying \eqref{consv} .
If there is no danger of confusion, 
we omit $V$ (resp. $E$, $A$) in $\varphi_V$ (resp. $\varphi_E$, $\varphi_A$).
By definition, $G^\prime_0 =  (\varphi^{-1}(V(G_0)), \varphi^{-1}(E(G_0)))$.
We define an operator 
$\mathscr{U}_0:\ell^2(V(G)) \to \ell^2(V(G^\prime))$ as 
\begin{equation}
\label{Vecid0} 
(\mathscr{U}_0\psi)(x^\prime) 
= \begin{cases}
	\psi(x), & x^\prime = \varphi^{-1}(x) \quad  (x \in V(G_0)) \\
	0, & \mbox{otherwise} 
	\end{cases}
\end{equation}
for $\psi \in \ell^2(V(G))$.
Since, in general, ${\rm deg}_G x \not= {\rm deg}_{G^\prime} \varphi^{-1}(x)$,
we cannot hope that $\mathscr{U}_0$ is partial isometric.
As shown in the example below, 
it can also be the case that $\mathscr{U}_0 L_G \not= L_{G^\prime} \mathscr{U}_0$.
\begin{example}[Lattice with pendants]
\label{ex_pendant}
{\rm
Let $G=\mathbb{Z}$ be the one-dimensional lattice 
and $G^\prime = (V(G^\prime), E(G^\prime))$ be defined 
by $V(G^\prime) = \mathbb{Z} \times \{0,1\}$
and 
\begin{align*} 
E(G^\prime) 
& =  \{ e \mid V(e) = \{(m,0), (n,0)\}, |m-n|=1 \} \\
& \quad \cup \{ e \mid V(e) = \{(m,0), (m,1)\} \}.
\end{align*}
A vertex of degree one is called a pendant vertex.
The vertices $(m, 1) \in V(G^\prime)$ ($m \in \mathbb{Z}$) are pendant vertices,
{\it i.e.}, ${\rm deg}_{G^\prime}(m,1) = 1$.
We set $G_0 = \mathbb{Z}$ and $G^\prime_0 := \{ (m,0) \mid m \in \mathbb{Z} \}$. We define a bijection $\varphi:V(G^\prime_0) \to V(G_0)$ as 
\[ \varphi((m,0)) = m, \quad (m,0) \in V(G^\prime_0). \] 
The graph $G^\prime$ is a perturbed graph of
the one-dimensional lattice $\mathbb{Z}$,
which is a graph obtained from $\mathbb{Z}$
by adding a pendant vertex to each vertex of $\mathbb{Z}$.
We know that ${\rm deg}_G m = 2 \not= 3 = {\rm deg}_{G^\prime} \varphi^{-1}(m)$
and
\begin{align*} 
\|\mathscr{U}_0\psi \|_{\ell^2(V(G^\prime))}^2
	& = \frac{3}{2} \|\psi\|_{\ell^2(V(G))}^2, 
		\quad \psi \in \ell^2(V(G_0)).
\end{align*}
By definition, $(\mathscr{U}_0\psi)(\cdot,1) \equiv 0$
and \begin{align*}
(\mathscr{U}_0L_G\psi)(\varphi^{-1}(m))
	& = \frac{3}{2} (L_{G^\prime} \mathscr{U}_0\psi)(\varphi^{-1}(m)),
		\quad \psi \in V(G).
\end{align*}
See \cite[p.3465]{Su}, where $G^\prime$ is denoted by $G_{1,1}$.  
The essential spectrum of $G^\prime$ is
\[ \sigma_{\rm ess}(L_{G^\prime}) 
	= \left[-1, -\frac{1}{3} \right] \cup \left[ \frac{1}{3}, 1 \right], \]
whereas $\sigma_{\rm ess}(L_\mathbb{Z})=[-1,1]$.

Next we consider the graph $G$ 
obtained from $\mathbb{Z}$ by adding pendant vertices  
to alternative vertices of $\mathbb{Z}$,
which was  studied in \cite{Su} and called $G_{2,1}$.
Let $G^\prime = G_{1,1}$ as above.
Then, $G^\prime$ is the perturbed graph of  $G$.
Indeed, we can check the condition in Definition \ref{def1.1}
as $G_0 = G_0^\prime = \mathbb{Z}$.
From \cite[p.3465]{Su}, we know that 
\[ \sigma_{\rm ess}(L_G) 
	= \left[-1, -\frac{1}{\sqrt{3}} \right] \cup \{0\} 
			\cup \left[ \frac{1}{\sqrt{3}}, 1 \right]. \]
In particular, we have 
$\sigma_{\rm ess}(L_{G}) \not \subset \sigma_{\rm ess}(L_{G^\prime})$
and $\sigma_{\rm ess}(L_{G^\prime}) \not \subset \sigma_{\rm ess}(L_{G})$.
}
\end{example}

In general, the restriction $\mathscr{U}_0 \mid_{\ell^2(V(G_0))}$
is not an isometry but an injection.
\begin{lemma}
{\rm
$\mathscr{U}_0 \mid_{\ell^2(V(G_0))}$
is an injection and 
\begin{equation} 
\label{U0rest}
c_0
\leq  \frac{\|\mathscr{U}_0 \psi\|_{\ell^2(V(G^\prime))}}{\|\psi\|_{\ell^2(V(G))}}
	\leq C_0, \quad \psi \in \ell^2(V(G_0))\setminus \{0\},
\end{equation}
where $c_0$ and $C_0$ are positive:
\[ c_0 = \frac{\inf_{x \in V(G_0)} {\rm deg}_{G^\prime} \varphi^{-1}(x)}{\sup_{x \in V(G_0)} {\rm deg}_G x},
	\quad C_0 = \frac{\sup_{x \in V(G_0)} {\rm deg}_{G^\prime}\varphi^{-1}(x)}{\inf_{x \in V(G_0)} x}. \]
}
\end{lemma}
\begin{proof}
{\rm
It suffices to prove \eqref{U0rest}.
From \eqref{deg>=1}, we know that $0 < c_0 \leq C_0 < \infty$. 
By direct calculation, 
\begin{align*} 
\|\mathscr{U}_0 \psi \|_{\ell^2(V(G^\prime))}^2
	& = \sum_{x \in V(G_0)} |\psi(x)|^2{\rm deg}_{G^\prime}\varphi^{-1}(x) \\
	& = \sum_{x \in V(G_0)} |\psi(x)|^2{\rm deg}_G x \frac{{\rm deg}_{G^\prime}\varphi^{-1}(x)}{{\rm deg}_G x},
\end{align*}
which yields \eqref{U0rest}.
}
\end{proof}

To obtain the result that 
$\sigma_{\rm ess}(L_G) \subset \sigma_{\rm ess}(L_{G^\prime})$,
we must divide the unperturbed part of $G^\prime$ that preserves 
the graph structure of $G$ from the perturbed part. 
To this end, we set
\[ \Lambda 
	= \{ x \in V(G_0) \mid 
		 {\rm deg}_{G^\prime} \varphi^{-1}(x) = {\rm deg}_G x,~
		 	 A_x(G) \subset A(G_0) \}.
\]
\begin{lemma}
\label{lem_bulk}
{\rm
Let $x \in \Lambda$.
Then,
\begin{equation}
\label{int_edge} 
\varphi^{-1}(A_x(G)) = A_{\varphi^{-1}(x)}(G^\prime).
\end{equation}
Moreover, it follows that for all $\psi \in \ell^2(V(G))$,
\begin{equation} 
\label{intertwine}
(\mathscr{U}_0 L_{G}\psi)(\varphi^{-1}(x))
 = (L_{G^\prime}\mathscr{U}_0\psi)(\varphi^{-1}(x)). 
\end{equation}
}
\end{lemma}
\begin{proof}
Because, by definition, $A_x(G) \subset A(G_0)$,
we have  $A_x(G) = A_x(G_0)$.
Combining this with $G_0 \simeq G^\prime_0$
yields the result that
\begin{equation}
\label{eq_bulk01} 
\varphi^{-1}(A_x(G)) = A_{\varphi^{-1}(x)}(G^\prime_0). 
\end{equation}
To show \eqref{int_edge}, it suffices to prove
\begin{equation}
\label{eq_bulk02}
A_{\varphi^{-1}(x)}(G^\prime_0) = A_{\varphi^{-1}(x)}(G^\prime).
\end{equation}
Clearly, $A_{\varphi^{-1}(x)}(G^\prime_0) \subset A_{\varphi^{-1}(x)}(G^\prime)$.
We prove the equality. 
Because $G^\prime_0 \simeq G_0$ and $A_x(G) = A_x(G_0)$, 
${\rm deg}_{G^\prime_0} \varphi^{-1}(x) = {\rm deg}_{G_0} x$
and ${\rm deg}_{G} x = {\rm deg}_{G_0} x$.
Suppose that 
$A_{\varphi^{-1}(x)}(G^\prime) \setminus A_{\varphi^{-1}(x)}(G^\prime_0) 
\not= \emptyset$. 
Then, 
\[ {\rm deg}_{G^\prime} \varphi^{-1}(x)
	>  {\rm deg}_{G^\prime_0} \varphi^{-1}(x) 
	=  {\rm deg}_{G_0} x 
	=  {\rm deg}_{G} x. \]
This contradicts $x \in \Lambda$,
because, by the definition of $\Lambda$,
${\rm deg}_{G^\prime} \varphi^{-1}(x) = {\rm deg}_G x$. 
This proves \eqref{eq_bulk02} and hence \eqref{int_edge}.

By \eqref{vecon}, \eqref{Vecid0}, and \eqref{int_edge}, 
\begin{align*} 
(\mathscr{U}_0 L_{G}\psi)(\varphi^{-1}(x))
&  = \frac{1}{{\rm deg}_G x} 
	\sum_{e \in A_x(G)} \psi(t(e))
	\\
& = \frac{1}{{\rm deg}_{G^\prime} \varphi^{-1}(x)} 
	\sum_{\varphi^{-1}(e) \in A_{\varphi^{-1}(x)}(G^\prime)} 
		(\mathscr{U}_0\psi)(t(\varphi^{-1}(e))) \\
& = (L_{G^\prime}\mathscr{U}_0\psi)(\varphi^{-1}(x)) ,
		\quad x \in \Lambda.
\end{align*}
This proves \eqref{intertwine}.
\end{proof}
\begin{remark}
{\rm
We use the condition 
\begin{equation}
\label{eq_bulk00}
A_x(G) \subset A(G_0)
\end{equation} 
in the definition of $\Lambda$ 
to prove Lemma \ref{lem_bulk}. 
The condition \eqref{eq_bulk00} does not hold in general,
even if  ${\rm deg}_{G^\prime} \varphi^{-1}(x) = {\rm deg}_G x$.
See Example \ref{ex_cone}. 
The vertex $x = (x_1,0) \in V(G_0)$ satisfies
${\rm deg}_G x = {\rm deg}_{G^\prime} \varphi^{-1}(x)$.
However, \eqref{eq_bulk00} does not hold,
because $A_x(G) \setminus A_x(G_0) \not=\emptyset$. 
}
\end{remark}
Let $P_{\varphi^{-1}(\Lambda)} : \ell^2(V(G^\prime)) \to  \ell^2(V(G^\prime))$ 
be the orthogonal projection onto 
the closed subspace
\[ \ell^2(\varphi^{-1}(\Lambda)) 
	:= \{ \psi^\prime \in \ell^2(V(G^\prime)) \mid 
		{\rm supp}\psi^\prime \subset \varphi^{-1}(\Lambda) \} \]
and $P_{\varphi^{-1}(\Lambda)}^\perp := 1 - P_{\varphi^{-1}(\Lambda)}$.
Because by \eqref{intertwine},
$P_{\varphi^{-1}(\Lambda)} (L_{G^\prime}\mathscr{U}_0 -\mathscr{U}_0 L_G) = 0$,
\begin{equation}
\label{perturb}
L_{G^\prime}\mathscr{U}_0
	= \mathscr{U}_0 L_G + K_{\Lambda}, 
\end{equation}
where 
\begin{equation} 
\label{K}
K_{\Lambda}
	:= P_{\varphi^{-1}(\Lambda)}^\perp (L_{G^\prime}\mathscr{U}_0 -\mathscr{U}_0 L_G). 
\end{equation}
In this sense, we say that $\varphi^{-1}(\Lambda)$ 
(resp. $\varphi^{-1}(\Lambda)^{\rm c}$) is 
the unperturbed (resp. perturbed) part of the perturbed graph $G^\prime$.
If $\# \varphi^{-1}(\Lambda)^{\rm c} < \infty$,
then $P_{\varphi^{-1}(\Lambda)}^\perp$ is a finite rank operator
and hence by \eqref{K}, $K_{\Lambda}$ is compact.
\begin{proposition}
\label{prop2.1}
{\rm
Let $G^\prime$ be a perturbed graph of $G$. 
If $\# \varphi^{-1}(\Lambda)^{\rm c} < \infty$,
then
\[ \sigma_{\rm ess}(L_{G}) \subset \sigma_{\rm ess}(L_{G^\prime}). \]
}
\end{proposition}
\begin{proof}
{\rm
Let $\lambda \in \sigma_{\rm ess}(L_{G})$ and $\{\psi_n\}$
be a Weyl sequence for $L_G$ such that
(i) $\lim_{n \to \infty} \|(L_G -\lambda )\psi_n\|=0$,
(ii) $\|\psi_n\| = 1$, and (iii) ${\rm w-}\lim_{n \to \infty} \psi_n = 0$. 
$\mathscr{U}_0\psi_n/ \|\mathscr{U}_0\psi_n\|$
is a Weyl sequence for $L_{G^\prime}$. 
Indeed, by \eqref{U0rest}, $\|\mathscr{U}_0\psi_n\| \geq c_0 > 0$
and hence  ${\rm w-}\lim_{n \to \infty}  \mathscr{U}_0\psi_n/\|\mathscr{U}_0\psi_n\| =0$.
By the compactness of $K_\Lambda$, it follows that
$\lim_{n \to \infty} \|(L_{G^\prime} - \lambda) (\mathscr{U}_0\psi_n/\|\mathscr{U}_0\psi_n\| )\|=0$.
Hence, $\lambda \in \sigma_{\rm ess}(L_{G^\prime})$.
}
\end{proof}

We want to prove $\sigma_{\rm ess}(L_{G}) \subset \sigma_{\rm ess}(L_{G^\prime})$
under the condition in which $K_\Lambda$ is allowed to be non compact, {\rm i.e.},
$\# \varphi^{-1}(\Lambda)^{\rm c} = \infty$.
This fact is established in Section \ref{sec.3} in the case in which $G$ is a periodic graph. 

\subsection{Periodic graph}
We end this section by providing the definition of $\mathbb{Z}^d$-periodic graphs,
which are not necessary contained in $\mathbb{R}^d$
(see \cite{AnIsMo, KoSa} for periodic graphs contained in $\mathbb{Z}^d$)
and allow multiple edges and loops.
We set $\mathbb{N}_s = \{ v_1, v_2,\dots, v_s\}$ ($s \in \mathbb{N}$) 
and define a translation on $\mathbb{Z}^d \times \mathbb{N}_s$ as
\[ \tau_a((m, v_i)) = (m - a, v_i), 
	\quad (m, v_i) \in \mathbb{Z}^d \times \mathbb{N}_s. \] 
We use $A_{u,v}(G)$ to denote the set of edges with $o(e) = u$, $t(e) = v$
($u, v \in V(G)$):
\[ A_{u,v}(G) = \{ e \in A_u(G) \mid t(e) = v \}. \]
\begin{definition}
{\rm
We say that a graph $G =(V(G), E(G))$
is a {\it  $\mathbb{Z}^d$-periodic graph} and write $G \in \mathscr{L}^d$ 
if $G$ is isomorphic to a locally finite graph $\Gamma = (V(\Gamma), E(\Gamma))$ 
satisfying ($\mathscr{L}_1$) - ($\mathscr{L}_2$).
\begin{itemize}
\item[($\mathscr{L}_1$)] There exists $s \in \mathbb{N}$ such that
	$V(\Gamma) = \mathbb{Z}^d \times \mathbb{N}_s$. 
\item[($\mathscr{L}_2$)]  For all $u, v \in V(\Gamma)$,
$\# A_{\tau_a(v),\tau_a(u)}(\Gamma) = \# A_{v,u}(\Gamma)$.
\end{itemize}
}
\end{definition}
The condition  ($\mathscr{L}_1$) ensures
that for a periodic graph $G \in \mathscr{L}^d$,
there is a bijection $\varphi:V(G) \ni x
\mapsto (m,v_i) \in \mathbb{Z}^d \times \mathbb{N}_s$:
\begin{equation}
\label{id} 
x =\varphi^{-1} (m,v_i). 
\end{equation}
We henceforth identify a vertex $x \in V(G)$ of a periodic graph
with $(m,v_i) \in \mathbb{Z}^d \times \mathbb{N}_s$ 
by \eqref{id}, and then write $x = (m, v_i)$. 
In this case, we write the vertex set of a periodic graph as
$V(G) \simeq \mathbb{Z}^d \times \mathbb{N}_s$.
By the relations \eqref{Vecid} and \eqref{Lid} (replacing $G^\prime$ with $G$
and $G$ with $\Gamma$), we also identify 
the Laplacian $L_{G}$ of a periodic graph with $L_\Gamma$.
Since, by ($\mathscr{L}_2$), ${\rm deg} \tau_a(x) = {\rm deg} x$ for $x \in V(G)$,
\begin{equation} 
\label{di}
d_i:= {\rm deg}(m, v_i) \quad (i=1,\ldots,s) 
\end{equation} 
are independent of $m \in \mathbb{Z}^d$.
The condition ($\mathscr{L}_2$) implies that 
there exists a graph automorphism $\varphi_{a, V}:V(G) \to V(G)$,
$\varphi_{a, E}:E(G) \to E(G)$ ($a \in \mathbb{Z}^d$) such that
$\varphi_{a,V}(x) = \tau_a(x)$ ($x \in V(G)$ and, if $V(e) = \{x,y\}$, then
$V(\varphi_{a,E}(e)) = \{ \tau_a(x), \tau_a(y) \}$.
We use the notation $\tau_a$ to denote the automorphism $\varphi_{a}$
(the subscripts $V$ and $E$ are omitted).

We define the translation $T_a$ on $\ell^2(V(G))$ ($a \in \mathbb{Z}^d$) as
\[ (T_a \psi)(x) = \psi(\tau_a(x)), \quad x \in V(G). \]
By ($\mathscr{L}_2$) again, 
the Laplacian $L_G$ commutes with $T_a$ 
for all $a \in \mathbb{Z}^d$, {\it i.e}., $[L_G, T_a] = 0$.
Hence, we expect that
$L_G$ and $T_a$ can be simultaneously decomposable. 
Indeed, this can be accomplished as follows (for details, see \cite{KoSa} and \cite{Su}).
Let $\ell^2(V_s)= \mathbb{C}^s$ be the Hilbert space with the inner product
\[ \langle \xi, \eta \rangle_{V_s} = \sum_{i=1}^s \bar{\xi}_i \eta_i d_i,
	\quad \xi, \eta \in \ell^2(V_s). \]
Let $\mathscr{F}:\ell^2(V(G)) \to \int_{\mathbb{T}^d} ^\oplus \ell^2(V_s) \frac{dk}{(2\pi)^d}$ 
be a unitary operator defined as
$(\mathscr{F} \psi) (k) 
	= \left( \hat\psi_i(k) \right)_{i=1}^s$, 
where $\hat\psi_i(k) = \sum_{m \in \mathbb{Z}^d} e^{-i k \cdot m}\psi(m,v_i)$. 
Then, we have 
\[ \mathscr{F}T_a \mathscr{F}^{-1} 
	= \int_{\mathbb{T}^d}^\oplus e^{-i ak} \frac{dk}{(2\pi)^d} \] 
and the Floquet-Bloch decomposition,
\begin{equation}
\label{FBD}
\mathscr{F}  L_G \mathscr{F}^{-1}
= \int_{\mathbb{T}^d} ^\oplus L_G(k) \frac{dk}{(2\pi)^d},
\end{equation} 
where $L_G(k)$ is the Floquet $s\times s$ matrix and
$k \in \mathbb{T}^d = \mathbb{R}^d/ (2\pi \mathbb{Z})^d$
is the quasimomentum.

Let $\pi_*$ and $\pi^*$ be projections on $V(G)$ defined by
\[ \pi_*(m,v_i) = m \in \mathbb{Z}^d, \quad \pi^*(m,v_i) = v_i \in \mathbb{N}_s,
	\quad (m,v_i) \in V(G) \]
and set
\[ A_{i,j}(G) = \{ e \in A_{(0,v_i)}(G) \mid \pi^*(t(e)) = v_j \}. \]
In their paper \cite{KoSa}, 
Korotyaev and Saburova introduced a convenient notation
called the {\it edge index} $\chi(e)$:
\[ \chi(e) = \pi_*(t(e)) - \pi_*(o(e)) \in \mathbb{Z}^d,
	\quad e \in A(G). \]
They called an edge $e$ with non-zero index a {\it bridge}.
By ($\mathscr{L}_2$) and \eqref{vecon},
$\chi$ is $\mathbb{Z}^d$-invariant, {\it i.e.},
\begin{align*} 
\chi(\tau_a(e)) 
	& = \pi_*(\tau_a(t(e))) - \pi_*(\tau_a(o(e))) \\
	& = (\pi_*(t(e))-a) - (\pi_*(o(e))-a) = \chi(e), 
		\quad a \in \mathbb{Z}^d. 
\end{align*}
We also have $\pi_*(o(\tau_{\pi_*(o(e))}(e))) = 0 \in \mathbb{Z}^d$ and
\[ \chi(e) = \pi_*(t(\tau_{\pi_*(o(e))}(e))), \quad e \in A(G). \]
In particular, if $e \in A_{i,j}(G)$ and $t(e) = (m,v_j)$,
then $\chi(e) = m$. 
The following are known:
\begin{lemma}[\cite{KoSa}, \cite{HiNo}]
\label{lem2.3}
{\rm
Let $G \in \mathscr{L}^d$. 
\begin{itemize}
\item[(i)] 
$L_G(k) = ((L_G)_{i,j}(k))_{i,j=1}^s$ 
in \eqref{FBD} is given by
\[ (L_G)_{i,j}(k) 
	=  \sum_{e \in A_{i,j}(G)}e^{i \chi(e) \cdot k}/d_i.
	 \]
\item[(ii)] $\displaystyle \sigma(L_G) = \sigma_{\rm ess}(L_G) = \bigcup_{i=1}^s \lambda_i(\mathbb{T}^d)$.
\end{itemize}
}
\end{lemma}
Using (i) of Lemma \ref{lem2.3}, we have the following.
\begin{proposition}
\label{prop2.4}
{\rm
Let $G \in \mathscr{L}^d$. 
\begin{equation*} 
(L_G \psi)(m,v_i) 
	= \sum_{j=1}^s \sum_{e \in A_{i,j}(G)} \psi(m + \chi(e), v_j)/d_i,
		\quad (m,v_i) \in V(G).
\end{equation*}
}
\end{proposition}
\begin{proof}
{\rm
By direct calculation, 
\begin{align*}
(L_G \psi)(m,v_i) 
	& = \int_{\mathbb{T}^d} \frac{dk}{(2\pi)^d}
		e^{ik \cdot m} \sum_{j=1}^s (L_G)_{i,j}(k)\hat \psi_j(k) \\
	& =  \sum_{j=1}^s \sum_{e \in A_{i,j}(G)}
		\int_{\mathbb{T}^d} \frac{dk}{(2\pi)^d}
		e^{ik \cdot (m +\chi(e)) } \hat \psi_j(k)/d_i \\
	&= \sum_{j=1}^s \sum_{e \in A_{i,j}(G)} \psi(m + \chi(e), v_j)/d_i.
\end{align*}
}
\end{proof}
\section{Results}\label{sec.3}

We call a perturbed graph $G^\prime$ of $G \in \mathscr{L}^d$
a {\it perturbed periodic graph} of $G$.
We define the propagation length  $l_G \in \mathbb{N}$ by
\begin{equation}
\label{proprang} 
 l_G = \sup_{j=1,\dots,d}~ \sup_{e \in A(G)}|\chi_j(e)|,
\end{equation}
where $\chi_j(e)$ is the $j$-th component of the edge index of $e$.
If $G$ is connected, then there exists a bridge, and hence  $l_G \geq 1$.  
We use the following condition.
\begin{itemize}
\item[($\mathscr{P}$)] There exists a sequence $\{x_n\}_{n=1}^\infty \subset V(G_0)$
such that 
\[ I_n(x_n) := \{ x \in V(G) \mid  \pi_*(x) - \pi_*(x_n) \in [-n-l_G+1,n+l_G-1]^d \} \subset \Lambda. \]
\end{itemize}
We are now in a position to state our main theorem:
\begin{theorem}
\label{mainthm}
{\rm
Let $G^\prime$ be a perturbed periodic graph of $G \in \mathscr{L}^d$. 
Suppose that $G^\prime$ satisfies ($\mathscr{P}$).
Then,
\[ \sigma_{\rm ess}(L_G) \subset \sigma_{\rm ess}(L_{G^\prime}). \] 
}
\end{theorem}
\begin{remark}
{\rm
From a physical point of view, 
$\sigma_{\rm ess}(L_G)$ 
can be considered as the bulk spectrum. 
As shown below,
the bulk spectrum $\lambda \in \sigma_{\rm ess}(L_G)$ 
corresponds to a Weyl sequence of states 
with support in the unperturbed part $\varphi^{-1}(\Lambda)$ of $G^\prime$.
}
\end{remark}

Since $\sigma(L_{G^\prime})$ is contained in $[-1,1]$,
we have the following.
\begin{corollary}
\label{FSP}
{\rm
Let $G^\prime$ be a perturbed periodic graph of $G \in \mathscr{L}^d$.
Suppose that $G$ satisfies ($\mathscr{P}$) and $G$ has the FSP.
Then, $G^\prime$ has the FSP.
}
\end{corollary}
\begin{proof}[Proof of Theorem \ref{mainthm}]
As in the proof of Proposition \ref{prop2.1},
it suffices to show the existence of a Weyl sequence for $L_{G^\prime}$.
To this end, we fix $\lambda \in \sigma_{\rm ess}(G)$.
By Lemma \ref{lem2.3}, 
we have $\lambda = \lambda_h(k_0)$ with some $h =1,\ldots, s$
and $k_0 \in \mathbb{T}^d$.
Let $\xi_0 \in \ell^2(V_s)$ be a normalized eigenvector of 
the Floquet matrix $L_G(k_0)$ corresponding to the eigenvalue $\lambda_h(k_0)$:
$L_G(k_0) \xi_0 = \lambda_h(k_0) \xi_0$.

For $n \in \mathbb{N}$, we define a function $\rho_n:\mathbb{Z}^d \to [0,1]$ as
\[ \rho_n(m) = \prod_{j=1}^d \rho(m_j/n), 
	\quad m = (m_1,\dots,m_d) \in \mathbb{Z}^d, \]
where 
\begin{align*}
  \rho(t) := \begin{cases}
		1 - |t|, \quad & |t|\leq 1, \\ 
		0,                          \quad & |t|> 1.
	       \end{cases}
\end{align*}
Then, $\rho_n$ is supported in $[-n+1,n-1]^d\cap \mathbb{Z}^d$ and 
\begin{equation}
\label{rhon}
\sum_{m \in \mathbb{Z}^d} |\rho_n(m)|^2
	= \prod_{j=1}^d \sum_{m_i \in [-n+1,n-1] \cap \mathbb{Z}}|\rho(m_j/n)|^2 
 	=  \left( \frac{ 2n^2+1}{3n} \right)^d.
\end{equation}
Let $\psi_n$ ($n \in \mathbb{N}$) be vectors in  $\ell^2(V(G))$ defined by
\[ {\psi_n}(m,v_i) = e^{ik_0 \cdot m}\rho_n(m) (\xi_0)_i,
	\quad (m,v_i) \in V(G), \]
where $(\xi_0)_i$ is the $i$-th component of $\xi_0$.
Because $\xi_0$ is a normalized vector, we know, from \eqref{rhon} that
\begin{equation}
\label{normofpsin01} 
\|\psi_n\|_{\ell^2(V(G))}^2
	= \left(\sum_{m \in \mathbb{Z}^d} |\rho_n(m)|^2 \right) \|\xi_0\|_{V_s}^2   
		= \left( \frac{ 2n^2+1}{3n} \right)^d.
\end{equation}
Combining ($\mathscr{P}$) with the fact  ${\rm supp}\rho_n = [-n+1,n-1]^d \cap \mathbb{Z}^d$,
we have
\begin{equation}
\label{supp01} 
{\rm supp} \psi_n(\tau_{\pi_*(x_n)}(\cdot)) \subset I_n(x_n) \subset \Lambda.
\end{equation}
Hence, by \eqref{U0rest} and \eqref{normofpsin01}, 
\begin{equation} 
\label{normofpsin02}
\|\mathscr{U}_0T_{\pi_*(x_n)}{\psi}_n\|
	\geq c_0 \left( \frac{ 2n^2+1}{3n} \right)^{d/2}.  
\end{equation}
We now define a sequence $\{\Psi_n\} \subset \ell^2(V(G^\prime))$ of normalized vectors as
\begin{equation} 
\label{Psin}
\Psi_n = \mathscr{U}_0T_{\pi_*(x_n)}\psi_n/\|\mathscr{U}_0 T_{\pi_*(x_n)} \psi_n\|. 
\end{equation}  
By definition, we know that $\sup_{x \in V(G)}|\psi_n(x)| \leq 1$.
By \eqref{normofpsin02},
\begin{equation*}
\sup_{x^\prime \in V(G^\prime)}
	|\Psi_n(x^\prime)| \leq {c_0}^{-1} \left( \frac{ 2n^2+1}{3n} \right)^{-d/2}.
\end{equation*}
Hence, it follows that
\[ \lim_{n \to \infty} \langle \Phi, \Psi_n \rangle = 0 \]
for	all finitely supported vectors $\Phi \in V(G^\prime)$, {\it i.e.}, 
$\#{\rm supp}\Phi < \infty$.
A standard limiting argument yields the result that ${\rm w-}\lim_{n \to \infty} \Psi_n = 0$.

It remains to prove the following:
\begin{equation}
\lim_{n \to \infty} (L_{G^\prime} - \lambda) \Psi_n = 0.
\end{equation}
By \eqref{perturb}, we observe that
\begin{equation}
\label{LGprimePsin}
(L_{G^\prime}-\lambda)\Psi_n
	= C_n \left(\mathscr{U}_0 (L_G - \lambda) T_{\pi_*(x_n)} \psi_n
		+ K_{\Lambda} T_{\pi_*(x_n)} \psi_n\right),
\end{equation}
where $C_n := \|\mathscr{U}_0T_{\pi_*(x_n)}{\psi}_n\|^{-1}$.
Since $T_a$ commutes with $L_G$ for all $a \in \mathbb{Z}^d$,
the first term of \eqref{LGprimePsin} is
\begin{equation}
\label{eq*01}
C_n \mathscr{U}_0 (L_G-\lambda) T_{\pi_*(x_n)} \psi_n
	= C_n \mathscr{U}_0T_{\pi_*(x_n)}  (L_G -\lambda) \psi_n. 
\end{equation}
	
We will prove that the second term of \eqref{LGprimePsin} vanishes.
Because by  the definition of $K_\Lambda$, 
$(K_{\Lambda} T_{\pi_*(x_n)} \psi_n)\mid_{\varphi^{-1}(\Lambda) }= 0$,
it suffices to prove the following:
\begin{equation}
\label{*02} 
(K_{\Lambda} T_{\pi_*(x_n)} \psi_n)\mid_{\varphi^{-1}(\Lambda)^{\rm c} }= 0. 
\end{equation}
Let $x^\prime \in \varphi^{-1}(\Lambda)^{\rm c}$.
By \eqref{K}, 
\begin{equation*}
\left(K_\Lambda T_{\pi_*(x_n)}\psi_n\right)(x^\prime)
= \left(L_{G^\prime}\mathcal{U}_0 T_{\pi_*(x_n)}\psi_n \right)(x^\prime) 
	- \left(\mathcal{U}_0 L_{G} T_{\pi_*(x_n)}\psi_n  \right)(x^\prime). 
\end{equation*}
To show \eqref{*02}, it suffices to prove the following lemma,
which is proved in the appendix.
\begin{lemma}
\label{lem3.5*}
{\rm
Let $x^\prime \in \varphi^{-1}(\Lambda)^{\rm c}$. 
\begin{itemize}
\item[(i)] $\left(L_{G^\prime}\mathcal{U}_0 T_{\pi_*(x_n)}\psi_n \right)(x^\prime)  =0$.
\item[(ii)] $\left(\mathcal{U}_0 L_{G} T_{\pi_*(x_n)}\psi_n  \right)(x^\prime) =0$.
\end{itemize}
}
\end{lemma}

Taking the above argument, \eqref{U0rest}, and eqref{normofpsin02} into account, 
we observe, from \eqref{LGprimePsin} and \eqref{eq*01} that
\begin{align} 
\|(L_{G^\prime}-\lambda)\Psi_n\|_{\ell^2(V(G^\prime))}
	&  = C_n \|\mathscr{U}_0T_{\pi_*(x_n)}(L_G-\lambda)\psi_n\|_{\ell^2(V(G^\prime))}  \notag \\
	& \leq   C_0 c_0^{-1} \left( \frac{ 2n^2+1}{3n} \right)^{-d/2} \|(L_G-\lambda)\psi_n\|_{\ell^2(V(G))}. \label{bound01}
\end{align}
Because $\xi_0$ is an eigenvector of $L_G(k_0)$ 
corresponding to $\lambda = \lambda_h(k_0)$,
it follows from Lemma \ref{lem2.3} that
\begin{align*}
\lambda\psi_n(m,v_i)
	& = e^{ik_0 \cdot m} \rho_n(m) \lambda_h(k_0) (\xi_0)_i \\
	& = e^{ik_0 \cdot m} \rho_n(m) (L_G(k_0)\xi_0)_i \\
	& = \sum_{j=1}^s \sum_{e \in A_{i,j}(G)} 
		e^{i(m + \chi(e)) \cdot k_0}\rho_n(m) (\xi_0)_j/d_i . 
\end{align*}
From Proposition \ref{prop2.4},
\begin{align*}
((L_G - \lambda)\psi_n) (m,v_i)
 	& = \sum_{j=1}^s \sum_{e \in B_{i,j}(G)}   
 		e^{i(m + \chi(e))\cdot k_0} (\xi_0)_j (\rho_n(m + \chi(e))  -\rho_n(m))/d_i,
\end{align*}
where $B_{i,j}(G)$ is the set of all bridges contained in $A_{i,j}(G)$:
\[ B_{i,j}(G) = \{ e \in A_{i,j}(G) \mid \chi(e) \not= 0 \}. \]
Let $B(G)$ be the set of all bridges
\[ B(G) = \{ e \in B_{i,j}(G) \mid i, j=1,\ldots, s \}. \]
By the Schwartz inequality and the fact that $d_i \geq 1$, 
\begin{align*}
\|(L_G - \lambda)\psi_n\|^2
	& = \sum_{m \in \mathbb{Z}^d} \sum_{i=1}^s|((L_G - \lambda)\psi_n) (m,v_i)|^2 \\
	& \leq  \sum_{m \in \mathbb{Z}^d} \sum_{i=1}^s
		\left( \sum_{j=1}^s \sum_{e \in B_{i,j}(G)}  (\xi_0)_j^2 \right) \\
	& \qquad \times 
		\left( \sum_{j=1}^s \sum_{e \in B_{i,j}(G)} 
			|\rho_n(m + \chi(e))  -\rho_n(m)|^2\right) \\
	& \leq \# B(G) \sum_{e \in B(G)} 
			\sum_{ m \in \mathbb{Z}^d} |\rho_n(m + \chi(e))  -\rho_n(m)|^2.
\end{align*}
Note that
\begin{align*}
|\rho_n(m + \chi(e))  -\rho_n(m)|
& = \prod_{\chi_j(e) =0} |\rho(m_j/n)| \\
& \qquad \times |\prod_{\chi_i(e) \not= 0} \rho((m_i-\chi_i(e))/n) - \prod_{\chi_i(e) \not=0} \rho(m_i/n)|
\end{align*}
and $\sum_{m \in \mathbb{Z}} |\rho(m/n)|^2 = \sum_{m \in \mathbb{Z}} |\rho((m-l)/n)|^2 $.
We observe that
\begin{align}
\|(L_G-\lambda)\psi_n\|^2
	& \leq \#B(G) \left(\sum_{m \in \mathbb{Z}^d} |\rho(m)|^2 \right)^{d-1} \notag \\
	& \qquad \times
		 \sum_{e \in B(G)} \sum_{\chi_i(e) \not=0}
		 	 \sum_{m \in \mathbb{Z}^d}
		 	 	|\rho((m - \chi_i(e))/n) -\rho(m/n)|^2. \label{bound02}
\end{align}
Combining \eqref{bound01} with \eqref{bound02} yields the result that
\begin{align*}
\|((L_{G^\prime})-\lambda)\Psi_n\|_{\ell^2(V(G^\prime))}^2
& \leq C_0^2 \# B(G) \left( \frac{ 2n^2+1}{3n} \right)^{-1} \\
& \quad \times \sum_{e \in B(G)} \sum_{\chi_i(e) \not=0}
		 	 \sum_{m \in \mathbb{Z}^d}
		 	 	|\rho((m - \chi_i(e))/n) -\rho(m/n)|^2.
\end{align*}
We complete the proof of Theorem \ref{mainthm} 
by using the following lemma.
\end{proof}
\begin{lemma}
\label{lem3.3}
{\rm
For all $l \in \mathbb{Z}$,
\[ \sum_{m \in \mathbb{Z}} |\rho((m - l)/n) - \rho(m/n)|^2 = O(n^{-1}) \]
as $n \to \infty$.
}
\end{lemma}
We prove the lemma in the appendix.

\section{Examples}
In this section, we present some examples of Theorem \ref{mainthm}
and Corollary \ref{FSP}.
\begin{example}(Random pendant graph)
\label{ex_random}
{\rm
We consider a graph obtained from $\mathbb{Z}^2$ 
by randomly adding pendant vertices.
Let $G = \mathbb{Z}^2$. 
Note that $G \in \mathscr{L}^d$,
because $G \simeq \mathbb{Z}^2 \times\{v_1\}$
and $ \mathbb{Z}^2 \times\{v_1\}$ satisfies ($\mathscr{L}_1$) and ($\mathscr{L}_2$).
Then, the map $\pi_*$ is trivial, {\it i.e.}, $\pi_*(x)=x$ for all $x\in\mathbb{Z}^2$.

Let $q_x$ $(x\in\mathbb{Z}^2)$ be a Bernoulli independent, identically distribution (i.i.d.) random variable
 with $\mathbb{P}(q_x=1) = \mathbb{P}(q_x=0) = \frac{1}{2}$.
We define 
\begin{align*}
   G_0     &:= G \\
   V(G_0') &:= \{(x,0) \mid x\in \mathbb{Z}^2 \},  \\
   E(G_0') &:= \{ e \mid V(e) = \{(x,0),(y,0)\} , |x-y|=1 \},  \\
   V(G')   &:= V(G_0') \cup \{(x,1) \mid q_x=1\} \subset \mathbb{Z}^2 \times\{0,1\}, \\
   E(G')   &:=   E(G_0') \cup \{ e \mid V(e) = \{(x,0),(x,1)\}, q_x=1 \}.
\end{align*}
Then, $G'=(V(G'),E(G'))$ is a perturbed graph of $G$ with 
$\varphi((x,0)) = x$ $(x \in V(G_0))$. 
The vertex $(x,1) \in V(G^\prime)$ is a pendant vertex,
which is added to the vertex $(x,0)$ of $G^\prime_0 \simeq \mathbb{Z}^2$
with probability $\mathbb{P}(q_x=1) =1/2$. 
In this sense,
the graph $G^\prime$ is considered as 
a graph obtained from $\mathbb{Z}^2$ by adding pendant vertices with probability $1/2$.
See Fig. 1. 
\begin{figure}[tbp]
  \label{fig_random}
  \centering
  \input
  {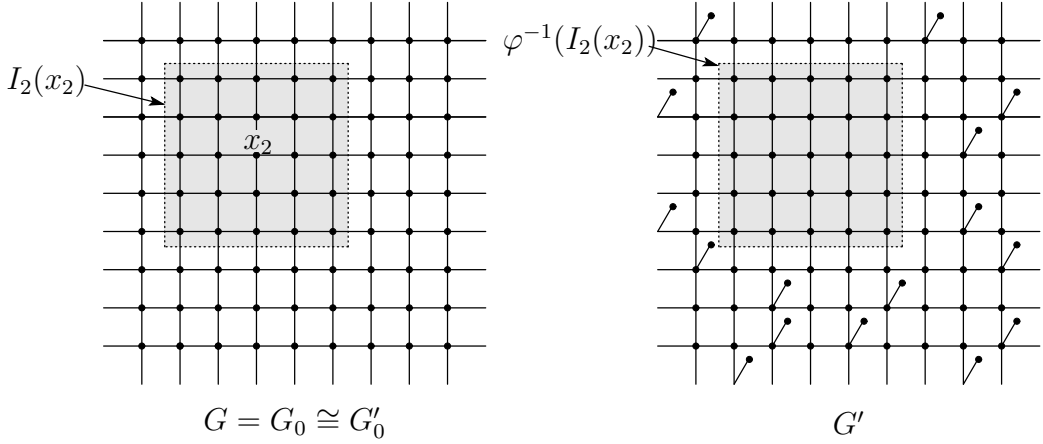}
  \caption{Graphs in Example \ref{ex_random}}
\end{figure}
In this case, $\ell_G = 1$ and $I_n(x) = \{y \in\mathbb{Z}^2 | y-x \in [-n,n]^2\}$.
Since $q_x$ is a Bernoulli i.i.d., 
$G^\prime$ satisfies ($\mathscr{P}$) almost surely.
Indeed, for each $x\in \mathbb{Z}^2$,
\begin{align*}
    \mathbb{P}(I_n(x) \not\subset \Lambda) 
    = 1 - \Big(\frac{1}{2}\Big)^{\# I_n(x)}
    = 1 - \Big(\frac{1}{2}\Big)^{(2n+1)^2}.
\end{align*}
Let $\xi_j \in \mathbb{Z}^2$ satisfy $I_n(\xi_j) \cap I_n(\xi_k) = \emptyset ~(i\neq j)$.
Then, for any $n\in \mathbb{N}$, 
\begin{align*}
   \mathbb{P}(\forall x \in \mathbb{Z}^2, ~ I_n(x) \not\subset \Lambda)
   ~\leq &~
   \mathbb{P}(\forall j=1,\dots,N, I_n(\xi_j) \not\subset \Lambda) \\
    = & \left[ 1- \left(\frac{1}{2} \right)^{(2n+1)^2} \right]^N 
    \longrightarrow  0 ~ (N \to \infty).
\end{align*}
Hence, there almost surely  exists
a sequence $\{x_n \} \subset \mathbb{Z}^2$ such that $I_n(x_n) \subset \Lambda$.
By Corollary \ref{FSP}, 
we have
\begin{align*}
\sigma_\mathrm{ess}(L_{G'}) = [-1,1], \quad \mathrm{a.s.}
\end{align*}
}
\end{example}

\begin{example}(Cone-like graph)
\label{ex_cone}
{\rm
The unperturbed graph is $G = \mathbb{Z}^2$.
We set
\begin{align*}
   V(G_0) &:= \{ x=(x_1,x_2) \mid x_i \geq 0, i=1,2 \},   \\
   E(G_0) &:= \{ e \mid V(e) = \{x,y\}, |x-y|=1, ~ x,y \in V(G_0)\}.
\end{align*}
A cone-like graph $G'=(V(G'),E(G'))$ is defined by
\begin{align*}
   V(G') &= V(G_0), \\
   E(G') &= E(G_0) \cup \{e \mid V(e) = \{(x_1,0),(0,x_1)\} \}.
\end{align*}
Setting $V(G_0') := V(G_0)$ and $E(G_0') := E(G_0)$, $G'$ becomes
 a perturbed graph of $\mathbb{Z}^2$.
It is easy to show that $\Lambda = \{(x_1,x_2) \in V(G') \mid x_1 \geq 1, x_2 \geq 1\}$.
As in the previous example, the maps $\pi_*$ and $\varphi$ are trivial and $\ell_G=1$.
Since $I_n(x):=\{y\in\mathbb{Z}^2 \mid y-x \in [-n,n]^2\}$, 
if we take $x_n = (n+1,n+1) \in V(G_0)$,
then $I_n(x_n) \subset \Lambda$. 
Thus, the graph $G'$ satisfies ($\mathscr{P}$).
Then, by Corollary \ref{FSP}, 
\begin{align*}
\sigma_\mathrm{ess}(L_{G'}) = [-1,1].
\end{align*}
}
\end{example}

\begin{example}(Upper-half plane)
\label{ex_upp}
{\rm
The unperturbed graph is $G = \mathbb{Z}^2$.
The upper-half plane is defined by
\begin{align*}
   V(G') &= \{(x_1,x_2) \mid x_1 \in \mathbb{Z}, ~ x_2 \geq 0 \}, \\
   E(G') &= \{ e \mid V(e) = \{x,y\}, |x-y|=1, ~ x,y \in V(G') \}.
\end{align*}
Let $G_0 = G_0^\prime = G^\prime$. 
$G'$ is a perturbed graph of $\mathbb{Z}^2$.
One can check that $\pi_* and \varphi$ are trivial maps, $\ell_G=1$, and
$\Lambda = \{(x_1,x_2) \in \mathbb{Z}^2 \mid x_2\geq 1\}$.
Setting $x_n = (0,n+1)$, $I_n(x_n) \subset \Lambda$.
By Corollary \ref{FSP},
we have
\begin{align*}
\sigma_\mathrm{ess}(L_{G'}) = [-1,1].
\end{align*}
}
\end{example}

\begin{example}
\label{counter_ex}
{\rm
Let $G \in \mathscr{L}^1$ be the graph obtained from $\mathbb{Z}$
by adding a pendant vertex to each $x \in \mathbb{Z}$ 
(see Example \ref{ex_pendant}). 
Let $G^\prime$ be a graph obtained from $G$
by adding a pendant vertex to each vertex $x \geq 0$.
$G^\prime$ is a graph obtained from $\mathbb{Z}$
by adding one pendant vertex to each vertex $x <0$
and two pendant vertices to each vertex $x \geq 0$.
See Fig. 2. 
\begin{figure}[htbp]
  \label{fig_counter}
  \centering
  \input
  {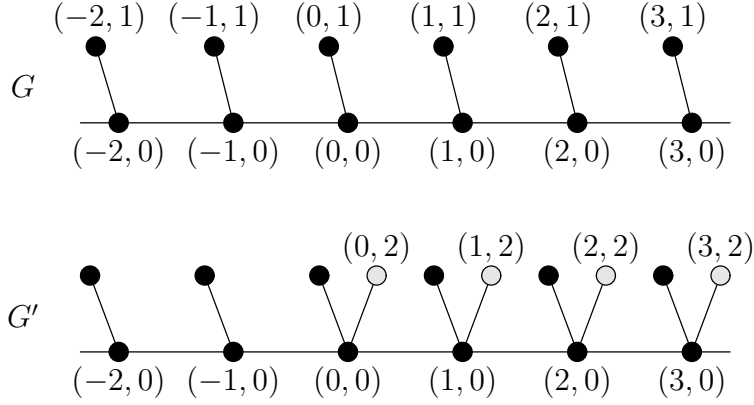}
  \caption{Graphs in Example \ref{counter_ex}}
\end{figure}
More precisely, we set
\begin{align*}
V(G^\prime) &= V(G) \cup \{(x,s) \mid x \geq 0, s \in \{ 0,1,2\} \}, \\
E(G^\prime) &= E(G) \cup \{ e \mid V(e) = \{ (x,0), (x,2) \}, x \geq 0 \}, \\
V(G) &= \{(x,s) \mid x \in \mathbb{Z}, s \in \{ 0,1\} \}, \\
E(G) & = \{ e \mid V(e) = \{(x,0), (y,0)\}, |x-y|=1 \} \\
	&\qquad \cup  \{ e \mid V(e) = \{(x,0), (x,1)\}, x \in \mathbb{Z} \}.
\end{align*} 

The vertices $(x,1)$ and $(x,2)$ are pendant vertices adjacent to $(x,0)$.
Let $G_0 = G_0^\prime = G \subset G^\prime$.
Then, $\Lambda = \{ (x,s) \in V(G) \mid x < 0, s \in \{0,1\} \}$.
Clearly, $G^\prime$ satisfies ($\mathscr{P}$). 
By Theorem \ref{mainthm}, we have
\[  \sigma_{\rm ess}(L_{G}) 
	= \left[ -1, -\frac{1}{3} \right] \cup \left[ \frac{1}{3}, 1 \right]
	\subset \sigma_{\rm ess}(L_{G^\prime}).  \]
On the other hand, $L_{G^\prime}$ has zero eigenvalues with infinite multiplicity.
Indeed, $\Psi^{(n)} \in \ell^2(V(G^\prime))$ ($n \geq 0$) defined below
are zero eigenstates:
\[ \Psi^{(n)}(x,s) 
	= \begin{cases}
		 1/\sqrt{2}, & (x,s) = (n,1) \\
		 -1/\sqrt{2}, & (x,s) = (n,2) \\
		 0, & {\rm otherwise}.
		 \end{cases}
\]
Thus, $0 \in \sigma_{\rm ess}(L_{G^\prime}) \setminus \sigma_{\rm ess}(L_G)$.
}
\end{example}
\appendix
\section{Proof of Lemmas}
\subsection{Proof of Lemma \ref{lem3.5*}}
We first prove (i) of Lemma \ref{lem3.5*}. 
Let $x^\prime \in \varphi^{-1}(\Lambda)^{\rm c}$.
From \eqref{Vecid0}, we have
\begin{align*} 
(L_{G^\prime}\mathscr{U}_0T_{\pi_*(x_n)}\psi_n)(x^\prime)
	& =\sum_{e^\prime \in A_{x^\prime}(G^\prime)} 
			\left(\mathscr{U}_0T_{\pi_*(x_n)}\psi_n\right)(t(e^\prime)) \\
	& =\sum_{e^\prime \in A_{x^\prime}(G^\prime), ~t(e^\prime) \in \varphi^{-1}(V(G_0))} 
			\left(T_{\pi_*(x_n)}\psi_n\right)(\varphi(t(e^\prime)))  \\
	& =\sum_{e^\prime \in A_{x^\prime}(G^\prime), ~t(e^\prime) \in \varphi^{-1}(V(G_0))} 
			\psi_n(\tau_{\pi_*(x_n)}(\varphi(t(e^\prime)))).
\end{align*}
It suffices to prove the following:
\begin{lemma}
{\rm
Let $x^\prime \in \varphi^{-1}(\Lambda)^{\rm c}$,  $e^\prime \in A_{x^\prime}(G^\prime)$
and $t(e^\prime) \in  \varphi^{-1}(V(G_0))$.
Then,
\[ \tau_{\pi_*(x_n)}(\varphi(t(e^\prime))) \in ({\rm supp}\psi_n) ^{\rm c}. \]
}
\end{lemma}
\begin{proof}
{\rm
Suppose that $\tau_{\pi_*(x_n)}(\varphi(t(e^\prime))) \in {\rm supp}\psi_n$. 
Then, $\pi_*(\varphi(t(e^\prime))) - \pi_*(x_n) \in [-n+1,n-1]^d$.
Hence, we know that $\varphi(t(e^\prime)) \in I_n(x_n) \subset \Lambda$ 
and $t(e^\prime) \in \varphi^{-1}(\Lambda)$.
By the definition of $\Lambda$, we have
${\rm deg}_{G^\prime} t(e^\prime)
	= {\rm deg}_G \varphi(t(e^\prime))$ 
and $A_{\varphi(t(e^\prime))}(G) \subset A(G_0)$
with the result that $A_{t(e^\prime)}(G^\prime) \subset A(G^\prime_0)$.
Thus, the inverse edge $\bar{e}^\prime$ of $e^\prime$ 
is also contained in $A(G^\prime_0)$ and 
has the origin $o(\bar{e}^\prime) =t(e^\prime)$
and terminal $t(\bar{e}^\prime) =o(e^\prime)$,
which are contained in $V(G_0^\prime)$.
Hence, $\varphi(e^\prime) \in A(G_0)$ abd
$x^\prime = o(e^\prime) \in V(G_0^\prime) \setminus \varphi^{-1}(\Lambda)$.
In particular, $\varphi(x^\prime) \not\in I(x_n)$
and $ \pi_*(\varphi(x^\prime)) - \pi_*(x_n) \not\in [-n-l_G+1, n + l_G-1]^d$.
Consequently,
\begin{align*} 
\chi(\varphi(e^\prime))
	& = \pi_*(t(\varphi(e^\prime)) - \pi_*(o(\varphi(e^\prime)) \\
	& = \pi_*(\varphi(t(e^\prime)) - \pi_*(\varphi(x^\prime)) \\
	& = (\pi_*(\varphi(t(e^\prime))  - \pi_*(x_n))
		- ( \pi_*(\varphi(x^\prime)) - \pi_*(x_n)),
\end{align*}
which yields $|\chi_j(\varphi(e^\prime))| \geq l_G + 1$.
This is a contradiction, and therefore
$\tau_{\pi_*(x_n)}(\varphi(t(e^\prime))) \not\in {\rm supp}\psi_n$. 
}
\end{proof}

In what follows, we prove (ii) of Lemma \ref{lem3.5*}.
Because, by \eqref{Vecid0}, 
\[ \left(\mathcal{U}_0 L_G T_{\pi_*(x_n)} \psi_n \right) (x^\prime)
	= \begin{cases}
		\left(  L_G T_{\pi_*(x_n)} \psi_n \right) (x^\prime), 
			& x^\prime \in \varphi^{-1}(V(G_0)) \\
		0,
			& \mbox{otherwise},
		\end{cases} \]
it suffices to show that
\begin{equation}
\label{eq*001} 
\left( L_G T_{\pi_*(x_n)} \psi_n \right) (x^\prime) = 0 
\end{equation}
for all $x^\prime = \varphi^{-1}(x)$ with $x \in V(G_0) \cap \Lambda^{\rm c}$.
By the definition of $L_G$, 
\begin{align*}
\left( L_G T_{\pi_*(x_n)} \psi_n \right) (x^\prime)
& = \sum_{e \in A_x(G)} \left( T_{\pi_*(x_n)} \psi_n \right) (t(e)) \\
& = \sum_{e \in A_x(G)} \psi_n (\tau_{\pi_*(x_n)} (t(e))).
\end{align*}
To show \eqref{eq*001}, we need only to prove the following lemma.
\begin{lemma}
{\rm
Let $x \in V(G_0) \cap \Lambda^{\rm c}$ and $e \in A_x(G)$.
Then,
\[ \tau_{\pi_*(x_n)} (t(e)) \in  ({\rm supp}\psi_n) ^{\rm c}. \]
}
\end{lemma}
\begin{proof}
{\rm
Let $\tau_{\pi_*(x_n)} (t(e)) \in  {\rm supp}\psi_n$.
Then,
\[ \pi_*(t(e)) - \pi_*(x_n) \in [-n+1, n-1]^d, \]
which implies $t(e) \in I_n(x_n) \subset \Lambda$. 
Since $x \not\in \Lambda$, $x \not\in I_n(x_n)$ and
\begin{align*}
\chi(e) & = \pi_*(t(e)) -\pi_*(o(e)) \\
	& = \pi_*(t(e)) - \pi_*(x) \\
	& = (\pi_*(t(e)) - \pi_*(x_n))
		+ (\pi_*(x_n) - \pi_*(x)).
\end{align*}
Because $\pi_*(t(e)) - \pi_*(x_n) \in [-n+1, n-1]^d$
and $\pi_*(x_n) - \pi_*(x) \not\in [-n+l_G;1, n+l_G-1]^d$,
$|\chi_j(e)| \geq l_G+1$. 
This is a contradiction,
and therefore $\tau_{\pi_*(x_n)} (t(e)) \not\in  {\rm supp}\psi_n$.
}
\end{proof}

\subsection{Proof of Lemma \ref{lem3.3}}
Let 
\[ f_l(n) = \sum_{m \in \mathbb{Z}} |\rho((m - l)/n) - \rho(m/n)|^2. \]
Because $\rho$ is an even function, 
\begin{align*} 
f_{-l}(n) 
	& = \sum_{m \in \mathbb{Z}} |\rho((m + l)/n) - \rho(m/n)|^2 \\
	& = \sum_{m \in \mathbb{Z}} |\rho((-m - l)/n) - \rho(-m/n)|^2 
		= f_l(n).
\end{align*}
Hence, we can assume that $l \in \mathbb{N}$ and $n>l$ without loss of generality.
By direct calculation, 
\[ f_l(n) = I_1 + I_2 + I_3, \]
where
\begin{align*} 
& I_1 := \sum_{m =-n, \ldots,n-l}
			|\rho((m+l)/n) -\rho(m/n)|^2 \\
& I_2 := \sum_{m =-n-l +1,\ldots,-n} |\rho((m+l)/n)|^2\\
& I_3 := \sum_{m =n-l,\ldots,n-1} |\rho(m/n)|^2.
\end{align*}
We first estimate $I_1$. By the triangle inequality, 
\begin{align*}
I_1\leq  & =  n^{-2} \sum_{m =-n, \ldots,n-l}
			||m+l| - |m||^2 \\
	& \leq (l/n)^2 \sum_{m=-n, \ldots,n-l}1
		= \left( 2-\frac{l+1}{n} \right) \frac{l^2}{n} = O(n^{-1}).
\end{align*}
Next we show that $I_2 = I_3 = O(n^{-2})$.
Because $\rho$ is even, 
the change of coordinates $m^\prime = m+l$ yields
\begin{align*}
I_2& = \sum_{m^\prime =-n +1,\ldots,-n+l} |\rho(m^\prime/n)|^2 \\
	& = I_3 	
		= n^{-2} \sum_{m =n-l,\ldots,n-1} |n - m|^2
		= n^{-2} \sum_{k= 1}^l k^2 =O(n^{-2}).
\end{align*}
This proves Lemma \ref{lem3.3}.

\section*{Acknowledgments}
The authors thank Professor Fumihiko Nakano
for providing useful comments. This work was supported by Grant-in-Aid for Young Scientists (B) (No. 26800054).


\end{document}